\documentclass[a4paper,USenglish]{article}
\usepackage{microtype,xspace,multirow}

\title{Applications of incidence bounds in point covering problems
\footnote{Research funded by MADALGO, Center for Massive Data Algorithmics, a Center of the Danish National Research Foundation, grant DNRF84}
}

\author{Peyman Afshani,
    Edvin Berglin,
    Ingo van Duijn,
    Jesper Sindahl Nielsen\\
    \texttt{\{peyman,berglin,ivd,jasn\}@cs.au.dk}}

\usepackage{fullpage}
\usepackage{amsthm,amsmath, amssymb, }
\usepackage{todonotes}
\usepackage{algorithm}
\usepackage{algorithmicx}
\usepackage{mathtools}
\usepackage[noend]{algpseudocode}

\usepackage{mathrsfs}

\newcommand{\rthm}[1]{Theorem~\ref{thm:#1}}

\newcommand{\rlem}[1]{Lemma~\ref{lem:#1}}

\newcommand{\robs}[1]{Observation~\ref{obs:#1}}
\newcommand{\rsec}[1]{Section~\ref{sec:#1}}

\newcommand{\ostar}[1]{\mathcal{O}^*\!\left(#1\right)}
\newcommand{\BigO}[2][\mathcal{O}]{#1\!\left(#2\right)}
\newcommand{\uni}{\mathcal{U}}
\newcommand{\tuple}[1]{\left\langle#1\right\rangle}
\newcommand{\set}[1]{\left\{#1\right\} }
\newcommand{\ST}{Szemer\'edi-Trotter}
\newcommand{\reals}{\mathbb{R}}
\newcommand{\algo}[1]{\textsc{#1}}
\newcommand{\prob}[1]{\textsl{#1}}

\newcommand{\C}{\mathscr{C}}
\newcommand{\LL}{\mathcal{L}}
\newcommand{\ie}{i.e.\ }
\newcommand{\eg}{e.g.\ }
\newcommand{\para}[1]{\vspace{2mm} \noindent\textbf{#1}}

\newtheorem{theorem}{Theorem}
\newtheorem{observation}[theorem]{Observation}
\newtheorem{definition}[theorem]{Definition}
\newtheorem{lemma}[theorem]{Lemma}

\newcommand{\thmheadfont}{\textcolor{darkgray}{$\blacktriangleright$}\nobreakspace\sffamily\bfseries}
\newenvironment{repeatenv}[2]%
  {\medskip \noindent {\thmheadfont #1~\ref{#2}.}\ \slshape}
  {\normalfont}

\newenvironment{repeatlemma}[1]{\begin{repeatenv}{Lemma}{#1}}      {\end{repeatenv}}
\renewcommand{\epsilon}{\varepsilon}

\begin{document}
  
  \maketitle

\begin{abstract}
  In the \prob{Line Cover} problem a set of $n$ points is given and the task is to cover the points using either the minimum number of lines or at most $k$ lines.
  In \prob{Curve Cover}, a generalization of \prob{Line Cover}, 
  the task is to cover the points using curves with $d$ degrees of freedom.
  Another generalization is the \prob{Hyperplane Cover} problem where points in $d$-dimensional space are to be covered by hyperplanes.
  All these problems have kernels of polynomial size, where the parameter is the minimum number of lines, curves, or hyperplanes needed.

  First we give a non-parameterized algorithm for both problems in $\ostar{2^n}$ (where the $\ostar{\cdot}$ notation hides polynomial factors of $n$) time and polynomial space, beating a previous exponential-space result.
  Combining this with incidence bounds similar to the famous Szemer\'edi-Trotter bound, 
  we present a \prob{Curve Cover} algorithm with running time $\ostar{(Ck/\log k)^{(d-1)k}}$, where $C$ is some constant.
  Our result improves the previous best times $\ostar{(k/1.35)^k}$ for \prob{Line Cover} (where $d=2$), $\ostar{k^{dk}}$ for general \prob{Curve Cover}, as well as a few other
  bounds for covering points by parabolas or conics. 
  We also present an algorithm for \prob{Hyperplane Cover} in $\reals^3$ with running time $\ostar{(Ck^2/\log^{1/5}\!k)^{k}}$, 
  improving on the previous time of $\ostar{(k^2/1.3)^{k}}$.
\end{abstract}

\newpage
\section{Introduction}
In the \prob{Line Cover} problem a set of points in $\reals^2$ is given and the task is to cover them using either the minimum number of lines, or at most $k$ lines where $k$ is given as a parameter in the input.
It is related to \prob{Minimum Bend Euclidean TSP} and has been studied in connection with facility location problems~\cite{estivill2011fpt,megiddo1982complexity}.
The \prob{Line Cover} problem is one of the few low-dimensional geometric problems that are known to be NP-complete~\cite{megiddo1982complexity}.
Furthermore \prob{Line Cover} is APX-hard, i.e.,\ it is NP-hard to approximate within factor $(1+\varepsilon)$ for arbitrarily small $\varepsilon$~\cite{kumar2000hardness}.
Although NP-hard, \prob{Line Cover} is fixed-parameter tractable when parameterized by its solution size $k$ so any solution that is ``not too large'' can be found quickly.

One generalization of the \prob{Line Cover} problem is the \prob{Hyperplane Cover} problem, where the task is to use the minimum number of hyperplanes to cover points in  $d$-dimensional space.
Another generalization is to cover points with algebraic curves, e.g.\ circles, ellipses, parabolas, or bounded degree polynomials.
These can be categorized as covering points in an arbitrary dimension space using algebraic curves with $d$ degrees of freedom and at most $s$ pairwise intersections.
We call this problem \prob{Curve Cover}.
The first parameterized algorithm that was presented for \prob{Line Cover} runs in time $\ostar{k^{2k}}$~\cite{langerman2005covering}.\footnote{Throughout the paper we use the 
$\ostar{\cdot}$
 notation to hide polynomial factors of a superpolynomial function.}
This algorithm generalizes to generic settings, such as \prob{Curve Cover} and \prob{Hyperplane Cover}, obtaining the
running time $\ostar{k^{dk}}$ where $d$ is the degree of the freedom of the curves or the dimension of the space for hyperplane cover.

The first improvement to the aforementioned generic algorithm reduced the running time to  $\ostar{(k/2.2)^{dk}}$ for the \prob{Line Cover} problem \cite{grantson2006covering}.
The best algorithm for the \prob{Hyperplane Cover} problem, including \prob{Line Cover}, runs in $\ostar{k^{(d-1)k}/1.3^k}$ time \cite{wang2010parameterized}.
A non-parameterized solution to \prob{Line Cover} using dynamic programming has been proposed with both time and space $\ostar{2^n}$ \cite{cao2012study}, which is time efficient when the number of points is $\BigO{k\log k}$.
Algorithms for parabola cover and conic cover appear in \cite{tiwari2012covering}, running in time $\ostar{(k/1.38)^{(d-1)k}}$ and $\ostar{(k/1.15)^{(d-1)k}}$ respectively.

\para{Incidence Bounds.}
Given an arrangement of $n$ points and $m$ lines, an \emph{incidence} is a point-line pair where the point lies on the line. Szemer\'edi and Trotter gave an asymptotic (tight) upper bound of $\BigO{(nm)^{2/3}+n+m}$ on the number of incidences in their seminal paper \cite{szemeredi1983extremal}.
This has inspired a long list of similar upper bounds for incidences between points and several types of varieties in different spaces, \eg ~\cite{elekes2005incidences,fox2014semi,pach1998number,solymosi2012incidence}.

\para{Our Results.} We give a non-parameterized algorithm solving the decision versions of both \prob{Curve Cover} and \prob{Hyperplane Cover} in $\ostar{2^n}$ time and polynomial space.
Furthermore we present parameterized algorithms for \prob{Curve Cover} and \prob{Plane Cover} (\prob{Hyperplane Cover} in $\reals^3$).
These solve \prob{Curve Cover} in time $\ostar{(Ck/\log k)^{(d-1)k}}$
and \prob{Plane Cover} in time $\ostar{(Ck^2/\log^{1/5}\!k)^{k}}$, both using polynomial space.
The main idea is to use Szemer\'edi-Trotter-type incidence bounds and using the aforementioned $\ostar{2^n}$ algorithm as a base case. We make heavy use of (specialized) incidence bounds and our running time is very sensitive to the maximum number of 
possible incidences between points and curves or hyperplanes.
In general, utilization of incidence bounds for constructing algorithms is rare (see \eg \cite{guibas1991exact,guibas1996exact}) and to our knowledge we are the first to do so for this type of covering problem.
It is generally believed that point sets that create large number of incidences must have some ``algebraic sub-structure'' (see e.g.\ \cite{GreenT13})
but curiously, the situation is not fully understood even in two dimensions.
So, it might be possible to get better specialized incidence bounds for us in the context of covering points.
Thus, we hope that this work can give further motivation to study specialized incidence bounds.

\section{Preliminaries}

\subsection{Definitions}
We begin by briefly explaining the concept of fixed-parameter tractability before formally stating the \prob{Curve Cover} and \prob{Hyperplane Cover} problems.

\begin{definition}A problem is said to be \emph{fixed-parameter tractable} if there is a parameter $k$ to an instance $I$, such that $I$ can be decided by an algorithm in time $\BigO{f(k)poly(|I|)}$ for some computable function $f$.
\end{definition}

The function $f$ is allowed to be any computable function, but for NP-complete problems can be expected to be at least single exponential. The name refers to the fact that these algorithms run in polynomial time when $k$ is (bounded by) a constant. Within the scope of this paper $I$ will typically be a set of points and $k$ is always a solution budget: the maximum allowed size of any solution of covering objects, but not necessarily the size of the optimal such solution.

Let $P$ be a set of $n$ points in any dimension, and $d,s$ be non-negative integers.
\begin{definition}
A set of algebraic curves $\C$ are called $(d,s)$-curves if
(i) any pair of curves from $\C$ intersect in at most $s$ points and (ii) for any $d$ points there are at most $s$ curves in $\C$ through them.
The parameter $d$ is the \emph{degrees of freedom} and $s$ is the \emph{multiplicity-type}.

\end{definition}
The set $\C$ could be an infinite set corresponding to a family of curves, and it is often defined implicitly.
We assume two geometric predicates: 
First, we assume that given two curves $c_1, c_2 \in \C$, we can find their intersecting points in polynomial time.
Second, we assume that given any set of up to $s+1$ points, in polynomial time, 
we can find a curve that passes through the points or decide that no such curve exists.
These two predicates are satisfied in the real RAM model of computation for many families of algebraic curves
and can be approximated reasonably well in practice.

We say that a curve \emph{covers} a point, or that a point is \emph{covered} by a curve, if the point lies on the curve. 
A set of curves $H \subset \C$ \emph{covers} a set of points $P$ if every point in $P$ is covered by a curve in $H$,
furthermore, $H$ is a $k$-cover if $|H|\leq k$.

\begin{definition}[Curve Cover Problem]
        Given a family of $(d,s)$-curves $\C$, a set of points $P$, and an integer $k$, does there exist a subset of $\C$ that is a $k$-cover of $P$?
\end{definition}

Now let $P$ be a set of points in $\reals^d$. A hyperplane covers a point if the point lies on the hyperplane. A set $H$ of hyperplanes covers a set of points if every point is covered by some hyperplane; $H$ is a $k$-cover if $|H|\leq k$.
In $\reals^d$, a $j$-flat is a $j$-dimensional affine subset of the space, e.g., $0$-flats are points, $1$-flats are lines and $(d-1)$-flats
are called hyperplanes.

\begin{definition}[Hyperplane Cover Problem]
Given an integer $k$ and a set $P$ of points in $\reals^d$, does there exist a set of hyperplanes that is a $k$-cover of $P$?
\end{definition}

For $d=3$ we call the problem \prob{Plane Cover}.
To make our parameterized \prob{Plane Cover} algorithm work, we need to introduce a third generalization: a version of \prob{Hyperplane Cover} where the input contains any type of flats. A hyperplane covers a $j$-flat for $j\leq d-2$ if the flat lies on the hyperplane; further notation follows naturally from the above.
\begin{definition}[Any-flat Hyperplane Cover Problem]
For $k \in \mathbb{N}$ and a tuple $P=\tuple{P_0,\dots,P_{d-2}}$, where $P_i$ is a set of $i$-flats in $\reals^d$, does there exist a set of hyperplanes that is a $k$-cover of $P$?
\end{definition}
We stress that our non-parameterized algorithm in \rsec{IE} solves \prob{Any-flat Hyperplane Cover} while the parameterized algorithm in \rsec{planecover} solves \prob{Plane Cover}.
\prob{Line Cover} is a special case of both  \prob{Curve Cover} and \prob{Hyperplane Cover}.
Since \prob{Line Cover} is known to be both NP-hard~\cite{megiddo1982complexity} and APX-hard~\cite{kumar2000hardness}, the same applies to its three generalizations as well. 

\subsection{Kernels}
Central to parameterized complexity theory is the concept of \emph{polynomial kernels}. A parameterized problem has a polynomial kernel if an instance $\tuple{P,k}$ in polynomial time can be reduced to an instance $\tuple{P',k'}$ where $|P'|$ and $k'$ are bounded by polynomial functions of $k$ and $\tuple{P,k}$ is a yes-instance if and only if $\tuple{P',k'}$ is a yes-instance. Problems with polynomial kernels are immediately fixed-parameter tractable; simply run a brute force algorithm on the reduced instance.

\begin{lemma}
\label{lem:cckernel}
 For a family $\C$ of $(d,s)$-curves, \prob{Curve Cover} has a size $sk^2$ kernel where no curve in $\C$ covers more than $sk$ points.
\end{lemma}
\begin{proof}
Suppose some curve $c \in \C$ covers at least $sk+1$ points in $P$.
These points cannot be covered by $k$ other curves $c_1, \dots, c_k$ as the pairwise intersection of $c_i$ with $c$ contains 
at most $s$ points.
Therefore every $k$-cover must include $c$; remove the points that it covers and decrement $k$.
We can repeat that for every curve that covers $sk+1$ points, until every curve in $\C$ covers at most $sk$ of the remaining points.
Thus, if the number of remaining points is more than $sk^2$, the instance has no $k$-cover and can be immediately rejected.
Otherwise, we are left with an instance of $sk^2$ points.
\end{proof}

For \prob{Any-flat Hyperplane Cover} a size $k^d$ kernel is presented in \cite{langerman2005covering}. It uses a \emph{grouping} operation, removing points and replacing them with higher dimension flats, which is not acceptable for a \prob{Hyperplane Cover} input. We present an alternative, slightly weaker hyperplane kernel containing only points; in $\reals^3$ it contains at most $k^3+k^2$ points.

\begin{lemma}
\label{lem:hpkernel}
Hyperplane Cover in $\reals^d, d\geq2,$ has a size $k^2(\sum_{i=0}^{d-2} k^i)=\BigO{k^d}$ kernel where for $j\leq d-2$ any $j$-flat covers at most $\sum_{i=0}^{j} k^i=\BigO{k^{j}}$ points and any hyperplane covers at most $k\sum_{i=0}^{d-2}=\BigO{k^{d-2}}$ points.
\end{lemma}
\begin{proof}
This kernel boils down to creating a maximum intersection $s=\sum_{i=0}^{d-2}k^i$ between two hyperplanes. Then we have a kernel of size $sk^2$ where no hyperplane covers more than $sk$ points, in exactly the same way as \rlem{cckernel}. For a point set $P$ in $\reals^d$, call a $t$-flat \emph{heavy} if it covers at least $\sum_{i=0}^tk^t$ points in $P$. $P$ is $t$-ready if every heavy $t$-flat covers exactly $\sum_{i=0}^tk^t$ points. When $P$ is $(d-2)$-ready, we have the desired intersection bound $s$ and are done. Clearly any point set is 0-ready, so we only need to show how to modify a $t$-ready point set into one that is equivalent and $(t+1)$-ready.

Let $P$ be $t$-ready and suppose it contains a heavy $(t+1)$-flat $f$. Since by assumption two $(t+1)$-flats intersect in at most $\sum_{i=0}^{t}k_i$ points, any hyperplane cover of $P$ must have a hyperplane covering $f$. This property is maintained by removing arbitrary points $R$ on $f$ so that $f$ covers exactly $\sum_{i=0}^{t+1}k_i$ points. Consider that some of the points $R$ were on another heavy flat $f_1$, which as a consequence is no longer heavy. This could mean that $P\setminus R$ has a cover where $P$ did not, and must be rectified. We do this by re-adding enough points to $f_1$ so that it too covers exactly $\sum_{i=0}^{t+1}k_i$. Put these new points in general position on $f_1$ to avoid creating new heavy $(t+1)$-flats or increasing the number of points on any heavy $t$-flat. Once all heavy $(t+1)$-flats are reduced to $\sum_{i=0}^{t+1}k_i$ points in this manner, the new point set is $(t+1)$-ready and a yes-instance if and only if $P$ is.
\end{proof}

Our algorithms will use both properties of the kernels. Kratsch et al.~\cite{kratsch2014point} showed that these kernels are essentially tight under standard assumptions in computational complexity.

\begin{theorem}[Kratsch et al. \cite{kratsch2014point}]
\prob{Line Cover} has no kernel of size $\BigO{k^{2-\epsilon}}$ 
unless coNP $\subseteq$ NP/poly.
\end{theorem}

\subsection{Incidence bounds}
Consider the \prob{Line Cover} problem.
Obviously, if the input of $n$ points are in general position, then we need $n/2$ lines to cover them. 
Thus, if $k \ll \frac{n}{2}$, we expect the points to contain ``some structure'' if they are to be covered by $k$ lines.
Such ``structures'' are very relevant to the study of incidences.
For a set $P$ of points and a set $L$ of lines, the classical \ST{} \cite{szemeredi1983extremal} theorem gives an upper bound on the number of point-line incidences, $I(L, P)$, in $\reals^2$.

\begin{theorem}[Szemer\'edi and Trotter \cite{szemeredi1983extremal}]
For a set $P$ of $n$ points and a set $L$ of $m$ lines in the plane, let $I(L,P) = |\set{(p,\ell) \mid p \in P \cap \ell, \ell \in L}|$. Then
$I(L, P) = \BigO{(nm)^{2/3} + n + m}$.
\end{theorem}

The linear terms in the theorem arise from the cases when there are very few lines compared to points (or vice versa).
In the setting of \prob{Line Cover} these cases are not interesting since they are easy to solve.
The remaining term is therefore the interesting one.
Since it it large, it implies there are many ways of placing a line such that it covers many points;
this demonstrates the importance of incidence bounds for covering problems.
We introduce specific incidence bounds for curves and hyperplanes in their relevant sections.

\section{Inclusion-exclusion algorithm}
\label{sec:IE}
This section outlines an algorithm \algo{Inclusion-Exclusion} that for both problems decides the size of the minimum cover, or the existence of a $k$-cover, of a point set $P$ in $\ostar{2^n}$ time and polynomial space. Our algorithm improves over the one from \cite{cao2012study} for \prob{Line Cover} which finds the cardinality of the smallest cover of $P$ with the same time bound but exponential space. The technique is an adaptation of the one presented in \cite{bjorklund2009set}; their paper immediately gives either $\ostar{3^n}$-time polynomial-space or $\ostar{2^n}$-time $\ostar{2^n}$-space algorithms for our problems. We give full details of the technique for completeness; to do so, we require the intersection version of the inclusion-exclusion principle.
\begin{theorem}[Folklore]
\label{thm:inclexcl}
Let $A_1, \dots, A_n$ be a number of subsets of a universe $\uni$. Using the notation that $\overline{A}=\uni\setminus A$ and $\bigcap_{i \in 
\emptyset} \overline{A_i} = \uni$, we have:
\begin{equation*}
\Bigg|\bigcap_{i \in \set{1,\dots,n}} A_i \Bigg| = \sum_{X \subseteq \set{1,\dots,n}} (-1)^{|X|} 
\Bigg|\bigcap_{i \in X} \overline{A_i}\Bigg|.
\end{equation*}
\end{theorem}

\subsection{Curve Cover}
\label{sec:cc-ie}
Let $P$ be the input set of points and $\C$ be the family
of $(d,s)$-curves under consideration. Although we are creating a non-parameterized algorithm, we nevertheless assume that we have access to the solution parameter $k$. This assumption will be removed later.
We say a set $Q$ is a \emph{coverable set in $P$} (or is \emph{coverable} in $P$) if $Q\subseteq P$ and $Q$ has a $1$-cover.

Let a \emph{tuple} (in $P$) be a $k$-tuple $\tuple{Q_1, \dots, Q_k}$ such that $\forall i: Q_i$ is coverable in $P$.
Note that there is no restriction on pairwise intersection between two coverable sets in a tuple. Define $\uni$ as the set of all tuples. For $p \in P$, let $A_p = \set{\tuple{Q_1,\dots,Q_k} \mid p \in \bigcup_i Q_i} \subseteq \uni$ be the set of all tuples where at least one coverable set contains $p$.

\begin{lemma}
\label{lem:numberofintersect}
$P$ has a $k$-cover if and only if $\left|\bigcap_{p\in P} A_p\right| \geq 1$.
\end{lemma}

\begin{proof}
Take a tuple in $\bigcap_{p\in P} A_p$. For each coverable set $Q$ in the tuple, place a curve that covers $Q$. Since the tuple was in the intersection, every point is in some coverable set, so every point is covered by a placed curve. Hence we have a $k$-cover.

Take a $k$-cover $\mathcal{C}$ and from each curve $c\in \mathcal{C}$ construct a coverable set of the points covered by $c$. Form a tuple out of these sets and observe that the tuple is in the intersection $\bigcap_{p\in P} A_p$, hence its cardinality is at least 1.
\end{proof}
Note that several tuples may correspond to the same $k$-cover, so this technique cannot be used for the counting version of the problem. \rthm{inclexcl} and \rlem{numberofintersect} reduce the problem of deciding the existence of $k$-covers to computing a quantity $\left|\bigcap_{i\in X} \overline{A_i}\right|$.
The key observation is that $\overline{A_p}$ is the set of tuples where no coverable set contains $p$ and $\bigcap_{i\in X} \overline{A_i}$ is the set of tuples that contain no point in $X$, \ie the set of tuples in $P\setminus X$. The remainder of this section shows how to compute the size of this set in polynomial time. 
Let $c(X)=\left|\set{Q \mid Q \subseteq X, Q \textrm{ is coverable in }X}\right|$ be the number of coverable sets in a point set $X$. A tuple in $P\setminus X$ is $k$ coverable sets drawn from a size $c(P\setminus X)$ pool (with replacement), hence there are $c(P\setminus X)^k$ such tuples. To compute $c(X)$ we introduce the notion of \emph{representatives}. Let $\pi$ be an arbitrary ordering of $P$. The representative $R=\set{r_1,\dots,r_i}$ of a coverable set $Q$ is the $\min(|Q|,s+1)$ first points in $Q$ as determined by the order $\pi$. Note that for any coverable set $Q$, it holds that $R\subseteq Q$.
Let $q(X,\pi,R)$ be the number of coverable sets that have the representative $R$.

\begin{lemma}
\label{lem:countrepresentative}
$q(X,\pi,R)$ can be computed in $\BigO{|X|}$ time and $\BigO{\log|X|}$ space.
\end{lemma}
\begin{proof}
If $R$ is not a valid representative, $q(X,\pi,R)=0$. If $|R|\leq s$, $q(X,\pi,R)=1$. If $|R|=s+1$, let $U$ be the union of every coverable set with representative $R$, and $X'=U\setminus R$.
The number of subsets of $X'$ is the number of coverable sets with representative $R$, \ie $q(X,\pi,R)=2^{|X'|}$. For any $p\in P$ with $\pi(p)>\pi(r_i)$, $p\in X'$ if and only if there is a curve $c\in \C$ such that $c$ covers $\set{r_1,\dots,r_i,p}$.
Since $i \leq s+1$ the time complexity is $\BigO{|X|}$.
The space complexity is logarithmic since we need only maintain $|X'|$ rather than $X'$.
\end{proof}

\begin{lemma}
\label{lem:countsets}
$c(X)$ can be computed in $\BigO{|X|^{s+2}}$ time and $\BigO{|X|}$ space.
\end{lemma}
\begin{proof}
Fix an ordering $\pi$. As every coverable set in $X$ has exactly one representative under $\pi$, we get that $c(X)=\sum_{R} q(X,\pi,R)$. There there are only $\BigO{\binom{|X|}{s+1}}=\BigO{|X|^{s+1}}$ choices of $R$ for which $q(X,\pi,R)>0$, and by \rlem{countrepresentative} each term of the sum is computable in $\BigO{|X|}$ time and logarithmic space. The space complexity is therefore dominated by the space to store $\pi$ which is linear.
\end{proof}

\begin{theorem}
\label{thm:cc-inclexcl}
There exists a $k$-cover of curves from $\C$ for $P$ if and only if
\begin{equation*}
\left|\bigcap _{p\in P} A_p\right| = \sum_{X \subseteq P} (-1)^{|X|} \left|\bigcap_{p \in X} \overline{A_p}\right| = \sum_{X \subseteq P} (-1)^{|X|} c(P \setminus X)^k \geq 1.
\end{equation*}
This comparison can be performed in $\BigO{2^n n^{s+2}}$ time and $\BigO{nk}$ bits of space.
\end{theorem}
\begin{proof}
 Since $c(P\setminus X)^k\leq 2^{nk}$ for any $X$ it can be stored in $nk$ bits. The absolute value of the partial sum can be kept smaller than $2^{nk}$ by choosing an appropriate next $X$. The rest follows from \rthm{inclexcl}, \rlem{numberofintersect} and \rlem{countsets}.
\end{proof}

Finally, we remove the assumption that we have the parameter $k$. Any input requires at most $n$ curves.
Since $k$ is only used to compute $c(X)^k$ we can try $k=1,2,\ldots,n$ and return the first $k$ with a positive sum.
This increases the time by an $\BigO{n}$ factor. Alternatively, we can run $n$ simultaneous sums, since the parameter $k$ is only accessed when computing $c(X)^k$. This increases the space by factor $\BigO{n}$ and the time by a lower-order additive term.

\subsection{Any-flat Hyperplane Cover}
Here we treat all flats in the instance $\tuple{P_0,\dots,P_{d-2}}$ as atomic objects and $P$ as a union $\bigcup_{i=0}^{d-2}P_i$.
This algorithm is very similar to that of \rsec{cc-ie}, so we only describe their differences. A set of flats $Q \subseteq P$ is a coverable set in $P$ if there exists a hyperplane that covers every $p \in Q$.
The representative of $\emptyset$ is $\emptyset$, and the representative of a non-empty coverable set $Q$ is a set $R=\set{r_1,\dots,r_i}$. Let $r_1$ be the first flat in $Q$ and for $j\geq 2$, $r_j$ is defined if the affine hull of $\set{r_1,\dots,r_{j-1}}$ has lower dimension than the affine hull of $Q$. If so, let $r_j$ be the first flat in $Q$ that is not covered by the affine hull of $\set{r_1,\dots,r_{j-1}}$.

\begin{lemma}
\label{lem:hpcountrepresentatives}
$q(X,\pi,R)$ can be computed in $\BigO{|X|}$ time and $\BigO{\log |X|}$ space.
\end{lemma}
\begin{proof}
If $R$ is not a valid representative, $q(X,\pi,R)=0$. Otherwise, let $U$ be the union all coverable sets with the representative $R$, and $X'=U\setminus S$.
For every $p\in X\setminus R$, let $j$ be the highest index such that $\pi(r_j)<\pi(p)$. Then $p \in X'$ if and only if $p$ is on the affine hull of $\set{r_1,\dots,r_j}$.
\end{proof}

There are $\BigO{\binom{n}{d}}$ representatives $R$ with $q(X,\pi,R)>0$ so the following two results hold; their proofs are analogous to \rlem{countsets} and \rthm{cc-inclexcl}.

\begin{lemma}
\label{lem:hpcountsets}
$c(X)$ may be computed in $\BigO{|X|^{d+1}}$ time and $\BigO{|X|}$ space.
\end{lemma}

\begin{theorem}
\label{thm:hp-inclexcl}
There exists a hyperplane $k$-cover for $P$ if and only if
\begin{equation*}
\left|\bigcap _{p\in P} A_p\right| = \sum_{X \subseteq P} (-1)^{|X|} \left|\bigcap_{p \in X} \overline{A_p}\right| = \sum_{X \subseteq P} (-1)^{|X|} c(P \setminus X)^k \geq 1.
\end{equation*}
This comparison may be performed in $\BigO{2^n n^{d+1}}$ time and $\BigO{nk}$ bits of space.
\end{theorem}

\section{Curve Cover}
\label{sec:curvecover}
Recall that we are considering $(d,s)$-curves, where $d$ and $s$ are constants.
Since we have a kernel of up to $sk^2$ points, \algo{Inclusion-Exclusion} used on its own runs in time $\ostar{2^{sk^2}}$ which is too slow to give an improvement.
We improve this by first using a technique that reduces the number of points in the input, and then using \algo{Inclusion-Exclusion}.
To describe this technique and the intuition behind it, we first provide a framework based on the following theorem by Pach and~Sharir.

\begin{theorem}[Pach and Sharir \cite{pach1998number}]
\label{thm:curves}
Let $P$ be a set of $n$ points and $L$ a set of $m$ $(d,s)$-curves in the plane. The number of point-curve incidences between $P$ and $L$ is
\[I(P,L) = \BigO{n^{d/(2d-1)}m^{(2d-2)/(2d-1)}+n+m}\]
\end{theorem}

Note that the above holds for curves in arbitrary dimension.
This can be seen by projecting the points and curves onto a random plane, which will keep the projection of distinct points, and
prevent the curves from projecting to overlapping curves.

\begin{definition}
Let a \emph{candidate} be any curve in $\C$ that covers at least 1 point in $P$. 
Define its \emph{richness} with respect to $P$ as the number of points it covers. A candidate is $\gamma$-rich if its richness is at least $\gamma$, and $\gamma$-poor if its richness is at most $\gamma$.
\end{definition}
Recall that from the kernelization in \rlem{cckernel}, it follows that every candidate is $sk$-poor. The following gives a bound on the number of $\gamma$-rich candidates.

\begin{lemma}
\label{lem:numberofcurves}
Let $P$ be a set of $n$ points in some finite dimension space $\reals^x$. The number of $\gamma$-rich candidates in $P$ is
$\BigO{\frac{n^d}{\gamma^{2d-1}}+\frac{n}{\gamma}}$
\end{lemma}
\begin{proof}
If $m$ curves pass through $\gamma$ or more points, this generates at least $m\gamma$ incidences. By Theorem~\ref{thm:curves} we get
\begin{equation*}
m\gamma = \BigO{n^{d/(2d-1)}m^{(2d-2)/(2d-1)}+n+m}.
\end{equation*}
We deal with the three terms in the $\BigO{\cdot}$ separately. If $m\gamma=\BigO{n^{d/(2d-1)}m^{(2d-2)/(2d-1)}}$, the expression simplifies to $m=\BigO{n^d\gamma^{-(2d-2)}}$. If $m\gamma=\BigO{n}$ we have that $m=\BigO{n/\gamma}$. If $m\gamma=\BigO{m}$ then $\gamma$ is a constant. Since at most $s$ curves pass through the same $d$ points, $m$ is bounded by the total number of distinct curves $s\binom{n}{d}=\BigO{n^d}$. Therefore, this case is covered by the first term, giving the total bound $m=\BigO{\frac{n^d}{\gamma^{2d-1}}+\frac{n}{\gamma}}$.
\end{proof}

\subparagraph{Intuition for algorithm.} We exploit the following observation: given a $k$-cover $\mathcal{C}$, some curves in $\mathcal{C}$ might be significantly richer than others.
The main idea of our technique is to try to select (\ie branch on) these rich curves first.
Since they cover ``many'' points, removing these decreases the ratio $|P|/k$ and calling \algo{Inclusion-Exclusion} eventually becomes viable.
The idea to branch on rich curves first has another important consequence.
Suppose we know that no candidate in $\C$ covers more than $\gamma$ points in $P$. 
This immediately implies that if there are strictly more than $k\gamma$ points in $P$, it is impossible to cover $P$.
Therefore we have $|P|/k\leq \gamma$.
Now look at the set of $\frac{\gamma}{2}$-rich candidates and decide for each whether to include it in the cover or not. 
By the earlier observation, including such a candidate is good for reducing the ratio $|P|/k$.
But excluding such a candidate has essentially the same effect, because that candidate will not be considered again (remove it from $\mathcal{C}$). Any remaining candidates in $\mathcal{C}$ now cover at most $\frac{\gamma}{2}$ points; we must have $|P|/k\leq\frac{\gamma}{2}$ (or the instance is not solvable) and have strengthened the bound on the ratio.
Regardless of which choice we make, we make progress towards being able to call the base case.

This strategy also makes sense from a combinatorial point of view, because from \rlem{numberofcurves} it follows that the search space is small for rich curves. Switching to \algo{Inclusion-Exclusion} early enough lets us bypass the potentially very large search space of poor candidates.

\subparagraph{The Algorithm.} Let $r$ be a parameter. The exact value is set in the proof of \rthm{cctime}, for now it is enough that $r=\BigO[\Theta]{\log k}$. For a budget $k$ let $\tuple{k_1,\dots,k_r}$ with $\sum_j k_j=k$ be a \emph{budget partition}. We describe a main recursive algorithm \algo{CC-Recursive} (see appendix for pseudocode) that takes 4 arguments:
the point set $P$, the class of curves $\C$, a budget partition $\tuple{k_1,\dots,k_r}$, and a recursion level $i$.
For convenience we define $\gamma_i = sk/2^i$.
A simple top-level procedure \algo{CurveCover} tries all budget partitions and calls the recursive algorithm with that partition at recursion level 1.

At every recursion depth $i$, let $K_i=\sum_{j=i}^r k_j$ be the remaining budget and $P_i$ the remaining point set.
That means earlier levels have created a partial solution $\mathcal{C}_{i-1}$ of $k-K_i$ curves covering the points $P \setminus P_i$.
The recursive algorithm will try to cover the remaining points using $\gamma_{i-1}$-poor curves.
Specifically, at depth $i$ let $S$ be the set of candidates from $\C$ that are $\gamma_{i}$-rich and $\gamma_{i-1}$-poor.
Since from depth $i$ and onward it has a remaining budget of $K_i$ and cannot pick candidates that are $(\gamma_{i-1}+1)$-rich, the algorithm rejects if strictly more than $K_i\gamma_{i-1}$ remain.
If fewer than $\frac{(d-1)}{2}K_i\log k$ points remain, the sub problem is solved with inclusion-exclusion.

If neither a reject (due to too many points) or a base-case call to inclusion-exclusion has occurred, the algorithm will branch. It does so in $\binom{|S|}{k_i}$ ways by simply trying all ways of choosing $k_i$ candidates from $S$. For each such choice, all points in $P$ covered by the chosen candidates are removed and the algorithm recurses to depth $i+1$. If all those branches fail, the instance is rejected.

\subsection{Analysis}
\begin{lemma}
\label{lem:cc-correct}
Algorithm \algo{CurveCover} decides whether $P$ has a $k$-cover of curves from $\C$.
\end{lemma}
\begin{proof}
Regard \algo{CurveCover} as being non-deterministic. Suppose $P$ has a $k$-cover $\mathcal{C}$. The proof is by induction on the recursion.
Assume as the induction hypothesis that the current partial solution $\mathcal{C}_{i-1}$ is a subset of $\mathcal{C}$ and that $\mathcal{C}$ contains no curves that are $(\gamma_{i-1}+1)$-rich when restricted to $P_i$.
The assumption is trivially true for $i=1$ as $\mathcal{C}_0=\emptyset$.

By the induction hypothesis, $\mathcal{C}\setminus\mathcal{C}_{i-1}$ is a $K_i$-cover for $P_i$ using only $\gamma_{i-1}$-poor curves.
Therefore it holds that $|P_i| \leq \gamma_{i-1}K_i$, and thus the algorithm does not reject incorrectly.
Furthermore, if \algo{Inclusion-Exclusion} is called it accepts since we are in the case that a solution exists.

Otherwise, let $D \subseteq \mathcal{C} \setminus \mathcal{C}_{i-1}$ be the curves that are $\gamma_{i}$-rich when restricted to $P_i$.
The algorithm non-deterministically picks $D$ from the set of candidates $S$ and constructs $\mathcal{C}_{i} = \mathcal{C}_{i-1} \cup D$.
This leaves $\mathcal{C}_i$ to be a subset of $\mathcal{C}$.
Additionally, $\mathcal{C}_i$ contains all $\gamma_i$-rich curves in $\mathcal{C}$ restricted to $P_i$ and hence to $P_{i+1}\subseteq P_i$, upholding the induction hypothesis.

Suppose the algorithm accepts the instance $\tuple{P,k}$. It can only accept if some call to \algo{Inclusion-Exclusion} accepts. Let $\mathcal{C}_r$ be the set of curves selected by the recursive part such that \algo{Inclusion-Exclusion} accepted the instance $\tuple{P\setminus \mathcal{C}_r,k-|\mathcal{C}_r|}$. Let $\mathcal{C}_{ie}$ be any $(k-|\mathcal{C}_r|)$-cover of $P \setminus \mathcal{C}_r$. Then $\mathcal{C}_r \cup \mathcal{C}_{ie}$ is a $k$-cover of $P$.
\end{proof}

By the nature of the inclusion-exclusion algorithm, \algo{CurveCover} detects the existence of a $k$-cover rather than producing one. But since \algo{CC-Recursive} produces a partial cover during its execution, it is straight-forward to extend that into a full $k$-cover by using \algo{Inclusion-Exclusion} as an oracle.

\subparagraph{Running time.}
To analyze the running time of the algorithm we see the execution of \algo{CC-Recursive} as a search tree $\mathcal{T}$.
Each leaf of the tree is either an immediate reject or a call to \algo{Inclusion-Exclusion}.
Since the latter is obviously most costly to run, we must assume for a worst case analysis that every leaf node calls the base case algorithm. The running time is the number of leaf nodes in the search tree times the running time of \algo{Inclusion-Exclusion}.
Since the algorithm performs exponential work in these leaf nodes but not in inner nodes, it is insufficient to reason about the \emph{size} of the tree. Therefore we will speak of the ``running time of a subtree'', which simply means the running time of the recursive call that corresponds to the root of that subtree.
We show that in the worst case, $\mathcal{T}$ is a complete tree $\mathcal{T}_1$ of depth $r$.
That is, $\mathcal{T}_1$ has no leaf nodes at depths less than $r$.

Let $\mathcal{T}_j$ be a complete subtree of $\mathcal{T}_1$ rooted at depth $j$.
To prove that $\mathcal{T}_1$ is the worst case for $\mathcal{T}$ we prove two things.
First we first prove an upper bound on the running time for arbitrary $\mathcal{T}_j$.
Then we prove that the running time of $\mathcal{T}_1$ can only improve if an arbitrary subtree is replaced by a leaf (\ie a call to \algo{Inclusion-Exclusion}).
The most involved part is proving an upper bound on the number of leaves of $\mathcal{T}_j$.

\begin{lemma}
\label{lem:leafbound}
Let $L$ be the number of leaves in $\mathcal{T}_j$.
Then for some constant $c_2=c_2(d,s)$, $L$ is bounded by
\[ L \leq \left(\frac{c_2k^d}{(k-k_r)\log^{d-1} k}\right)^{K_j-k_r}\]
\end{lemma}
The proof is long and tedious and we leave it for the appendix.
To give an idea of how \rlem{leafbound} is proved, we sketch a simplified worst case analysis for \prob{Line Cover}.
The analysis can be generalized to \prob{Curve Cover} and gives (up to a constant in the base of the exponent) the same running time as the real worst case.

\subparagraph{Analysis sketch.}
The branching of $\mathcal{T}_1$ at recursion level $i$ depends on the budget $k_i$ that is being used.
That means that the structure of the whole tree depends on the complete budget partition.
From \rlem{numberofcurves} it follows that the lower the richness the more candidates there are.
Since the richness halves after every recursive call, one could conjecture that the worst case budget partition would put as much budget in the end.
It could \eg look like $\tuple{0,0,\dots,0,k_{r-1},k_r}$, where $k-k_r = k_{r-1} > k_r$.
That is, only in the penultimate and last recursion level is there any budget to spend.
At the deepest level of recursion, the richness considered is strictly less than $\frac{\log k}{2}$ (because with this richness the base case algorithm is efficient).
Therefore, at the penultimate recursion level the richness is $\log k$.
At this level there are $k \log k$ points left and we can apply \rlem{numberofcurves} to bound the number of $\log k$ rich lines.
This yields a bound of $\frac{k^2}{\log k}$ on the number of candidates.
From these we pick $k-k_r$ lines, giving a branching of roughly $\left(\frac{k^2}{(k-k_r)\log k}\right)^{k-k_r}$ (where roughly means up to a constant in the base of the exponent).

It turns out that the worst case budget partition is in fact $\tuple{k_02^1, k_02^2,...,k_02^{r-1}, k_r}$ for some $k_0$.
However, to understand where the division by $\log^{d-1} k$ comes from in the expression of \rlem{leafbound}, it is sufficient to understand the above analysis sketch.
With \rlem{leafbound} in place, we can prove the following bound on the running time of $\mathcal{T}_j$.

\begin{lemma}
\label{lem:fullsubtree}
The time complexity of a complete subtree $\mathcal{T}_j$ is $\ostar{(c_4k/\log k)^{(d-1)K_j}}$, where $c_4=c_4(d,s)$ is a constant that depends on the family $\C$.
\end{lemma}
\begin{proof}
By \rlem{leafbound}, the number of leaves in
$\mathcal{T}_j$ is $L \leq {\left(\frac{c_2k^d}{(k-k_r)\log^{d-1}\!k}\right)}^{K_j-k_r}$.
Observe that at depth $r$, \algo{Inclusion-Exclusion} runs in time
$\ostar{2^{\frac{d-1}{2}k_r\log k}}=\ostar{(k^{1/2})^{(d-1)k_r}}$.
Since an inner node performs polynomial time work and the leaves perform exponential time work, this immediately implies that the running time for $\mathcal{T}_j$ is
\[\ostar{{\left(\frac{c_2k^d}{(k-k_r)\log^{d-1}\!k}\right)}^{K_j-k_r} \cdot (k^{1/2})^{(d-1)k_r}}.\]

Suppose that $k-k_r\geq c_3k$ for some constant $c_3>0$, then the running time solves to:
\begin{align*}
\ostar{\left(\frac{c_2k^{d-1}}{c_3\log^{d-1}\!k}\right)^{K_j-k_r} \cdot (k^{1/2})^{(d-1)k_r}}=\ostar{\left(\frac{c_4k}{\log k}\right)^{(d-1)K_j}}
\end{align*}
where $c_4 = (c_2/c_3)^{1/(d-1)}$.

When $k-k_r$ is less than a constant fraction of $k$, that is $k - k_r = \BigO[o]{k}$, it holds that $K_j - k_r = \BigO[o]{K_j}$ since $k \geq K_j \geq k_r$. 

\[\ostar{{\left(\frac{c_2k^d}{(k-k_r)\log^{d-1}\!k}\right)}^{o(K_j)} \cdot (k^{1/2})^{(d-1)(K_j-o(K_j))}}=\ostar{2^{o(dK_j\log k)+\frac{d-1}{2}(K_j\log k-o(k\log k))}}\]
With some simple algebra one gets that the exponent is bounded by $(d-1)K_j(\log k-\log\log k)$, giving the desired time bound $\ostar{2^{(d-1)K_j(\log k-\log\log k)}}=\ostar{(k/\log k)^{(d-1)K_j}}$.
\end{proof}

\begin{lemma}
\label{lem:innerleaf}
Let $L_j$ be a depth $j<r$ leaf of $\mathcal{T}$ that calls \algo{Inclusion-Exclusion}.
Then the running time of $\mathcal{T}_j$ dominates that of $L_j$.
\end{lemma}
\begin{proof}
By \rlem{fullsubtree}, the time complexity of $\mathcal{T}_j$ is $\ostar{(c_4k/\log k)^{(d-1)K_j}}$. 
At depth $j$ the algorithm has $K_j$ remaining budget to spend. Since the algorithm called \algo{Inclusion-Exclusion} at this depth, at most $\frac{d-1}{2}K_j\log k$ points remained and the call takes $\ostar{2^{\frac{d-1}{2}K_j\log k}}=\ostar{(k^{1/2})^{(d-1)K_j}}$ time, which is bounded by that for $\mathcal{T}_j$.
\end{proof}

\begin{theorem}
\label{thm:cctime}
\algo{CurveCover} decides \prob{Curve Cover} in time $\ostar{(Ck/\log k)^{(d-1)k}}$ where $C=C(d,s)$ is a constant that depends on the family $\C$.
\end{theorem}
\begin{proof}
Fix a budget partition $\tuple{k_1,\dots,k_r}$. By \rlem{innerleaf}, calling \algo{Inclusion-Exclusion} at a depth $j<r$ does not increase the running time of the algorithm. Therefore the time complexity of \algo{CC-Recursive} is $\ostar{(c_4k/\log k)^{(d-1)K_1}}={(c_4k/\log k)^{(d-1)k}}$.

\algo{CurveCover} runs \algo{CC-Recursive} over all possible budget partitions, of which by the ``stars and bars'' theorem are only $\binom{k+r-1}{k}$, a quasi-polynomial in $k$. Therefore by letting $C=c_4+\epsilon$ for any $\epsilon>0$, the time complexity of \algo{CurveCover} is $\ostar{(Ck/\log k)^{(d-1)k}}$.
\end{proof}

\begin{lemma}
\label{lem:ccpolyandspace}
The polynomial time dependency of \algo{CurveCover} is $\BigO{(k\log k)^{2+s}}$ and its space complexity is $\BigO{k^4\log^2k}$ bits.
\end{lemma}
\begin{proof}
The height of the tree is $r=\BigO{\log k}$. Inner nodes have polynomial time and space which is strictly dominated by the exponential time and polynomial space of the leaves. Hence the polynomial time dependency of \algo{CurveCover} is exactly the polynomial time dependency of the leaves. \algo{Inclusion-Exclusion} runs in $\BigO{n^{s+2}2^n}$ and $n=\BigO{k\log k}$ points remain when it is called; the polynomial dependency is $\BigO{(k\log k)^{s+2}}$.

\algo{Inclusion-Exclusion} requires only $\BigO{nk}=\BigO{k^2\log k}$ bits of storage, while an inner node stores its set of candidates $S$. A trivial bound on the size of any $S$ is $\binom{sk^2}{2}$ elements, which can be stored in $\BigO{k^4\log k}$ bits. Since $r=\BigO[\Theta]{\log k}$, we use no more than $\BigO{k^4\log^2k}$ bits to store them.
\end{proof}

\section{Hyperplane Cover}
\label{sec:planecover}
One generalization of \prob{Line Cover} was discussed in the previous section.
In this section we discuss its other generalization \prob{Hyperplane Cover}, and give an algorithm for the three dimensional case.
We would like to follow the same basic attack plan of using incidence bounds but here 
we face significant challenges and we need non-trivial changes in our approach.
One major challenge is the nature of incidences in higher dimensions. 
For example, the asymptotically maximum number of incidences between a set of points and hyperplanes in $d$-dimensions
is obtained by placing half of the points on one two-dimensional plane (see \cite{agarwal1992counting,Edelbook})
which clearly makes it an easy instance for our algorithm (due to kernelization).
Thus, in essence, we need to use specialized incidence bounds that disallow such configurations of points;
unfortunately, such bounds are more difficult to prove than ordinary incidence bounds
(and as it turns out, also more difficult to use).

\subsection{Point-Hyperplane incidence bounds in higher dimensions}

The most general bound for point-hyperplane incidences from \cite{agarwal1992counting,edelsbrunner1990complexity}
yields a bound of $\BigO[\Theta]{\frac{n^d}{\gamma^3}+\frac{n^{d-1}}{\gamma}}$ on the number of $\gamma$-rich hyperplanes in $d$ dimensions similar to \rlem{numberofcurves} (where the left term is again the significant one).
Our method requires that the exponent is greater in the denominator than in the numerator, so this bound is not usable beyond $\reals^2$.
As stated before, the constructions that make the upper bound tight are easy cases for our algorithm;
they contain very low dimensional flats that have many points on them.
A specialized bound appears in \cite{elekes2005incidences}, where the authors study the number of incidences between points and hyperplanes with a certain \emph{saturation}.

\begin{definition}
Consider a point set $P$ and a hyperplane $H$ in $\reals^d$. We say that $H$ is
$\sigma$-saturated, $\sigma > 0$, if $H \cap P$ spans at least $\sigma \cdot |H
\cap P|^{d-1}$ distinct $(d-2)$-flats of $F$.  
\end{definition}

For example in three dimensions, a $(1-\frac{1}{n})$-saturated plane contains no three collinear points.
The main theorem of \cite{elekes2005incidences} can be stated as follows.

\begin{theorem}[Elekes and T\'oth \cite{elekes2005incidences}]
\label{thm:sathyperplanes}
Let $d \geq 2$ be the dimension and $\sigma > 0$ a real number. There is a constant $C_1(d, \sigma)$ with the following property.
For every set $P$ of $n$ points in $\reals^d$, the number of $\gamma$-rich $\sigma$-saturated hyperplanes is at most:
\[
\BigO{C_1(d,\sigma) \left ( \frac{n^d}{\gamma^{d+1}} + \frac{n^{d-1}}{\gamma^{d-1}} \right)}.
\]
\end{theorem}
The interesting term in this bound has a greater exponent in the denominator, as required.
An issue is that it is not easy to verify if a hyperplane is $\sigma$-saturated.
In the same paper as \rthm{sathyperplanes}, the authors give another bound based on a more manageable property called \emph{degeneracy}.

\begin{definition}
Given a point set $P$ and a hyperplane $H$ in $\reals^d$,
we say that $H$ is $\delta$-degenerate, $0 < \delta \leq 1$, if $H \cap P$ is non-empty and at most $\delta \cdot |H \cap P|$ points of $H \cap P$ lie in any $(d-2)$-flat.
\end{definition}

For example in $\reals^3$, any $1$-degenerate plane might have all its points lying on a single line, and a plane with degeneracy strictly less than $1$ must have at least $3$ points not on the same line.
As such it is an easy property to test.

\begin{theorem}[Elekes and T\'oth \cite{elekes2005incidences}]
\label{thm:deghyperplanes}
For any set of $n$ points in $\reals^3$, the number of $\gamma$-rich $\delta$-degenerate planes is at most
\[
\BigO{\frac{1}{(1-\delta)^4} \left ( \frac{n^3}{\gamma^{4}} + \frac{n^{2}}{\gamma^{2}} \right)}.
\]
\end{theorem}
This bound is usable and relies on an easily-tested property, but unfortunately only applies to the $\reals^3$ setting.

\subsection{Algorithm for \prob{Plane Cover}}

In this section we present our algorithm \algo{PC-Recursive} that solves \prob{Plane Cover} using the bound from \rthm{deghyperplanes}.
This algorithm is similar to the algorithm for \prob{Curve Cover}, and it is assumed that the reader is sufficiently familiar with \algo{CC-Recursive} before reading this section.

Recall that by \rlem{hpkernel}, \prob{Plane Cover} has a kernel of size $k^3+k^2$ where no plane contains more than $k(k+1) \leq 2k^2$ points and no two planes pairwise intersect in more than $k+1$ points.
For convenience we define $\gamma_0=k^2+k$ and $\gamma_i = k^2/2^i$ for $i>0$.
We inherit the basic structure of the \algo{CC-Recursive} algorithm, such that every recursion level considers $\gamma_i$-rich-$\gamma_{i-1}$-poor candidates.
Additionally, any candidate considered must be \emph{not-too-degenerate}:

\begin{definition}
Let $\delta_i = 1-\gamma_i^{-1/5}$.
A $\gamma_i$-rich-$\gamma_{i-1}$-poor plane is called \emph{not-too-degenerate} if it is $\delta_i$-degenerate, and \emph{too-degenerate} otherwise.
\end{definition}
It is of no consequence that the definition does not cover all candidates considered on depth 1.
The main extension of \algo{PC-Recursive} compared to \algo{CC-Recursive} is to first use a different technique to deal with too-degenerate candidates, which then allows normal branching on the not-too-degenerate ones. The key observation is that any too-degenerate candidate has at least $\gamma_i\delta_i=\gamma_i-\gamma_i^{4/5}$ points on a line and at most $\gamma_{i-1}(1-\delta_i)=2\gamma_i^{4/5}$ points not on it.

Suppose a $k$-cover contains some too-degenerate plane $H$. By correctly guessing its very rich line $L$ and removing the points on the line, the algorithm makes decent progress in terms of shrinking the instance. The points on $H$ but not $L$ will remain in the instance even though the budget for covering them has been paid. These are called the \emph{ghost points} of $H$ (or of $L$), and $L$ is called a \emph{degenerate line}. The ghost points must be removed by extending the line $L$ into a full plane. But the ghost points are few enough that the algorithm can delay this action until a later recursion level. Specifically, for a line $L$ guessed at depth $i$, we extend $L$ into a plane at the first recursion depth which considers $2\gamma_i^{4/5}$-poor candidates, \ie the depth $j$ such that $\gamma_{j-1}\geq 2\gamma_i^{4/5}\geq \gamma_j$.

Therefore the algorithm keeps a separate structure $\LL$ of lines that have been guessed to be degenerate lines on some planes in the solution. Augment $\LL$ to remember the recursion depth that a line was added to it. At any recursion depth, the algorithm will deal with old-enough lines in $\LL$, then guess a new set of degenerate lines to add to $\LL$ before finally branching on not-too-degenerate planes.

\subparagraph{The algorithm} Let $r = \BigO[\Theta]{\log k}$ as before. Let $\tuple{h_1,\ell_1,\dots,h_r,\ell_r}$ with $\sum_{i=1}^rh_i+\ell_i = k$ be a budget partition.
The recursive algorithm \algo{PC-Recursive} takes 4 arguments: the point set $P$, a set of lines $\LL$, the budget partition, and a recursion level $i$.
A top level algorithm \algo{PlaneCover} tries all budget partitions and calls \algo{PC-Recursive} accordingly.

Let the current recursion depth be $i$, and let $K_i = \sum_{j=i}^r h_j + \ell_j$ be the remaining budget.  The sub-budget $h_i$ will be spent on not-too-degenerate planes, and $\ell_i$ on degenerate lines. Let $\LL$ be an augmented set of lines as described above.
This means that earlier levels have already created a partial solution of $k - (K_i + |\LL|)$ planes, and a set $\LL$ of lines that still need to be covered by a plane.
If strictly more than $(K_i+|\LL|)\gamma_{i-1}$ points remain, the algorithm rejects.
If at most $K_i\log k$ points, the algorithm switches to \algo{Inclusion-Exclusion}
passing on the instance $\tuple{P\cup\LL,K_i+|\LL|}$.

Let $f=\left\lceil\frac{5(i-1)-2\log k}{4}\right\rceil$. Let $A$ be the set of all lines in $\LL$ that were added at depth $f$ or earlier. Remove $A$ from $\LL$. For each way of placing $|A|$ planes $\mathcal{H}$ such that every plane contains one line in $A$ and at least one point in $P$, let $P'=P\setminus(P\cap \mathcal{H})$ be the point set not covered by these planes. For a $P'$, let $H$ be the set of not-too-degenerate planes and $L$ the set of degenerate lines too-degenerate candidates.

For every $P'$ and every way of choosing $h_i$ planes from $H$ and $\ell_i$ lines from $L$, branch depth $i+1$ by removing the covered points from $P$ and adding the chosen lines of $L$ to $\LL$.

\subsection{Analysis}
\subparagraph{Correctness.}To prove that the algorithm is correct, we follow a similar strategy as for \algo{CurveCover}.
We build on the notion that the algorithm is building up a partial solution of planes.
Removing the points covered by the partial solution yields a ``residual problem'' just as in \algo{CurveCover}.
A partial solution is \emph{correct} if it is a subset of some $k$-cover.
Correctness of the algorithm follows from proving that a $k$-cover exists if and only if one branch maintains a correct partial solution until it reaches \algo{Inclusion-Exclusion}.

The difference here is that the residual problem is an instance of \prob{Any-flat Plane Cover} and not \prob{Plane Cover}.
Therefore, we simply consider the original problem to be an instance of \prob{Any-flat Plane Cover}, namely $R_1=\tuple{P,\emptyset}$. We say that $\mathcal{C}$ covers $\tuple{P,\LL}$ if $\mathcal{C}$ covers both $P$ and $\LL$. What needs to be established is that there is a correct way to replace points with lines (\robs{makeline}) and, conversely, that there is a correct way to extend a line in $\LL_i$ (\robs{fixline}). The proofs for these are elementary and we omit them. Given these two facts, we can easily show that the algorithm will call \algo{Inclusion-Exclusion} on appropriate instances.

\begin{observation}
\label{obs:makeline}
Let $\ell$ be a line and $\mathcal{C}$ a set of planes such that some plane $h\in\mathcal{C}$ covers $\ell$. Then $\mathcal{C}$ is a cover for $\tuple{P,\LL}$ if and only if $\mathcal{C}$ is a cover for $\tuple{P\setminus\ell,\LL\cup\set{\ell}}$.
\end{observation}

\begin{observation}
\label{obs:fixline}
Let $\ell$ be a line, $\LL\ni\ell$ be a set of lines, and $\mathcal{C}$ be a set of planes such that some $h\in\mathcal{C}$ covers $\ell$ but not any other line $\ell'\in\LL$. Then $\mathcal{C}$ is a cover for $\tuple{P,\LL}$ if and only if $\mathcal{C}\setminus\set{h}$ is a cover of $\tuple{P\setminus h,\LL\setminus\set{\ell}}$.
\end{observation}
The conditions for \robs{fixline} might seem overly restrictive. But as the following \rlem{onedegline} shows, that situation arises when $\LL$ contains only correctly guessed degenerate lines.
\begin{lemma}
\label{lem:onedegline}
Let $h$ be a too-degenerate plane with degenerate line $\ell$ such that $h\setminus\ell$ is a too-degenerate plane with degenerate line $\ell'$. Then at no point during the execution of \algo{PC-Recursive} will $\LL$ contain $\ell$ and $\ell'$.
\end{lemma}
\begin{proof}
Let $j_1$ be the depth that $\ell$ was put in $\LL$, and $j_2$ the depth for $\ell'$. Since $\ell$ is a degenerate line, it
has at most $2\gamma_j^{4/5}$ ghost points. It gets removed from $\LL$ on some depth $i$ where $j_1\leq f=\left\lceil\frac{5(i-1)-2\log k}{4}\right\rceil$. Since $\ell'$ was put in $\LL$ before $\ell$ was taken out we have $j_2\leq i-1$ and $j_1 > \frac{5(j_2)-2\log k}{4}$. This implies that $\frac{k^2}{2^{j_2}} > \left(\frac{k^2}{2^{j_1}}\right)^{4/5}$ or $\gamma_{j_2}>\gamma_{j_1}^{4/5}$. Since $\ell_2$ is also a degenerate line, it was on a $\gamma_{j_2}$-rich candidate. This candidate contains more points than the possible number of ghost points of $\ell$, so $\ell$ was not a degenerate line.
\end{proof}
\begin{lemma}
\label{lem:ghostpoints}
If $\LL$ contains only the degenerate lines of some too-degenerate planes in a $k$-cover, the number of ghost points at depth $i$ is at most $|\LL|\gamma_{i-1}$.
\end{lemma}
\begin{proof}
Consider a degenerate line $\ell \in \LL$ guessed at some recursion depth $j<i$. Line $\ell$ was on a $\delta_j$-degenerate plane, \ie it covered at least $\gamma_j\delta_j$ points and has at most $2\gamma_j^{4/5}$ ghost points. Since $\ell$ was not removed from $\LL$ at depth $i-1$, we get $j>f_{i-1}=\left\lceil\frac{5((i-1)-1)-2\log k}{4}\right\rceil$. Some simple algebra gives $\frac{k^2}{2^{i-2}}\geq\left(\frac{k^2}{2^j}\right)^{4/5}$, \ie $\gamma_{i-1} \geq 2\gamma_j^{4/5}$. Hence any line in $\LL$ has left at most $\gamma_{i-1}$ ghost points in the instance, and the sum of ghost points is at most $|\LL|\gamma_{i-1}$.
\end{proof}
\begin{lemma}
\label{lem:pccorrect}
Algorithm \algo{PlaneCover} decides whether $P$ has a $k$-cover of planes.
\end{lemma}
\begin{proof}
View the algorithm as being non-deterministic. Suppose $P$ has a $k$-cover. \robs{makeline}, \robs{fixline} and \rlem{onedegline} guarantee that there is a correct path, and \rlem{ghostpoints} guarantees that the point set is not erroneously rejected. Therefore the algorithm will send a yes-instance to \algo{Inclusion-Exclusion} and accept.

Suppose $P$ has no $k$-cover. If the conditions for \robs{fixline} are not satisfied, removing $\ell$ from $\LL$ and pairing it up with points but not with another $\ell'\in\LL$ can only reduce the number of solutions. Therefore the algorithm detects no cover and rejects.
\end{proof}

We can now state our main theorem for \algo{Plane-Cover}.
\begin{theorem}
\label{thm:pcrunning}
\algo{PlaneCover} decides \prob{Plane Cover} in $\BigO{(Ck^2/\log^{1/5} k)^{k}}$ time for some constant $C$.
\end{theorem}

To give an idea of how to prove the above theorem, we give a sketch of the analysis that reflects the core of the real analysis.
As before, we assume a (slightly incorrect) worst case for the budget partition where all the budget is assigned to the two deepest recursion levels.
This gives a bound analogous to the bound appearing in \rlem{leafbound}.
After achieving this bound, the same arguments as for \algo{Curve-Cover} can be applied to achieve the bound from \rthm{pcrunning}.

\subparagraph{Analysis sketch.}
The branching of the analysis is twofold.
First there is the branching done on picking not-too-degenerate planes.
Secondly, we have the branching on too-degenerate planes.
This branching is actually a combination of picking the rich lines in too-degenerate planes, and the branching done by covering these lines with planes later on.

We sketch a bound here for the cases that either all the budget goes into picking not-too-degenerate planes, \emph{or} all budget goes into picking too-degenerate planes (\ie lines).
We show that if either (\textbf{i}) $\forall i, k_i = h_i$ or (\textbf{ii}) $\forall i, \ell_i = k_i$, then the branching can be bounded by $\left(\frac{k^3}{(k-k_r)\log ^{1/5} k}\right)^{k-k_r}$ (compare to \rlem{leafbound}).
The full proof for \rthm{pcrunning} shows that if the budget is distributed between these cases, then taking the product of the worst case running times of both cases is roughly the same as what we present here.
For both cases, we again assume a (slightly incorrect) worst case budget partition where $k_{r-1} + k_r = k$ and $k_{r-1} \geq k_r$.
By the same arguments as in the analysis sketch of \rsec{curvecover} we have the following two parameters at recursion level $r-1$;
the number of of points remaining is $n = k \log k$ and the richness $\gamma_{r-1}$ is $\log k$.

For (\textbf{i}) we can directly apply \rthm{deghyperplanes} as follows.
For the term $\frac{1}{(1-\delta)^4}$ we can substitute $\delta$ with $1-\gamma_{r-1}^{-1/5} = 1 - \log^{-1/5} k$ to get $\log ^{4/5} k$.
Plugging in all these values in \rthm{deghyperplanes} we get that the number of not-too-degenerate planes is bounded by $\frac{\log^{4/5}n^3}{\log^{4} k} = \frac{k^3}{\log^{1/5} k}$.
From these candidates we pick $k - k_r$ planes, giving a branching of $\binom{\frac{k^3}{\log^{1/5} k}}{k-k_r}$, which is roughly $\left(\frac{k^3}{(k-k_r)\log ^{1/5} k}\right)^{k-k_r}$.

For (\textbf{ii}) we do the following.
The algorithm picks $\gamma_{i+1}$-rich lines at level $i$, and these lines are matched with points at later level $j$ where $\gamma_j = \gamma_i^{4/5}$.
The cost for branching at level $j$ is charged to level $i$, so that we can more easily analyze the total branching on lines selected at level $i$.
With the budget partition as stated above, we can now bound the branching done at level $r-1$.
By the \ST{} theorem, there are at most $\frac{n^2}{(\log k)^3} = \frac{k^2}{\log k}$ candidates, from which we select $k-k_r$ lines.
This yields a total branching of $\binom{\frac{k^2}{\log k}}{k-k_r}$, which is roughly $\left(\frac{k^2}{(k-k_r)\log k}\right)^{k-k_r}$.
We then need to mach these $k-k_r$ lines with $k \log^{4/5} k$ points, yielding a further branching of $(k \log^{4/5} k)^{k-k_r}$.
Taking the product of both these branching factors gives $\left(\frac{k^2}{(k-k_r)\log k}\right)^{k-k_r} \cdot (k \log^{4/5} k)^{k-k_r} = \left(\frac{k^3}{(k-k_r)\log^{1/5} k}\right)^{k-k_r}$.

\begin{proof}
It holds that $\sum_{i=1}^r h_i + \sum_{i=1}^r \ell_i = \sum_{i=1}^r k_i = k - K_r$, so for convenience we define $\epsilon$ such that $\epsilon (k - K_r) = \sum_{i=1}^r h_i$.

By the bound we have that the number of not-too-degenerate planes at level $i$ is:
\begin{align*}
\left(\frac{1}{1-\delta_i}\right)^4 \frac{(k\gamma_{i-1})^3}{(\gamma_i)^4} = \BigO{ \frac{k^3}{\gamma_i^{1/5}}} = \BigO{2^{i/5}k^{13/5}}
\end{align*}

For convenient notation, set $\alpha = 2^{1/5}$ and $\beta = c_1 k^{13/5}$ for some constant $c_1$, so that the above expression becomes $\alpha^i\beta$.
The total branching for picking planes at level $i$ can thus be bounded by $\binom{\alpha^i \beta}{h_i}$.
Taking the product of branching factors at each level gives the following (very similar to curve case):

\[\prod_{i=1}^r \binom{\alpha^i \beta}{h_i} \leq \prod_{i=1}^r \left( \frac{\alpha^i \beta}{\alpha^i h_0}\right)^{\alpha^i h_0} = \prod_{i=1}^r \left( \frac{\beta}{h_0}\right)^{\alpha^i h_0}
= \left( \frac{\beta}{h_0}\right)^{\sum_{i=1}^r h_i} = \left( \frac{\beta}{h_0}\right)^{\epsilon (k - K_r)}\]

The number of $\gamma_{i+1}$-rich lines at level $i$ is $\frac{(k \gamma_{i-1})^2}{\gamma_{i+1}^3} = \frac{2k^2}{\gamma_i}$.
From these we pick $\ell_i$ lines, giving a branching of $(\frac{k^2}{\gamma_i\ell_i})^{\ell_i}$ (up to constants in the base).
The number of points at level $j$ where $\gamma_j = \gamma_i^{4/5}$ is $k\gamma_i^{4/5}$, thus matching $\ell_i$ lines with this many points yields a branching of $(k\gamma_i^{4/5})^{\ell_i}$.
Combining this with the branching factor above gives
\begin{align*}
 (k\gamma_i^{4/5})^{\ell_i} \cdot (\frac{k^2}{\gamma_i\ell_i})^{\ell_i} = \left ( \frac{\alpha^i \beta}{\ell_i}\right)^{\ell_i}
\end{align*}
By the same technique as the curves and planes, the total branching on lines can thus be bounded by
\[
\left( \frac{\beta}{\ell_0}\right)^{\sum_{i=1}^r \ell_i} = \left( \frac{\beta}{\ell_0}\right)^{(1-\epsilon) (k - K_r)}
\]

Similar to the way the value $h_0$ is lower bounded in \rsec{longproof} in the Appendix, the numbers $h_0$ and $\ell_0$ can be lower bounded by $\BigO[\Omega]{\frac{\epsilon (k - K_r)\log ^{1/5} k}{k^{2/5}}}$ and $\BigO[\Omega]{\frac{(1-\epsilon) (k - K_r)\log ^{1/5} k}{k^{2/5}}}$.
Let $c_2$ be the constant that collects the implicit constant in these lower bounds, the ignored constants in the base, and $c_1$.
The total branching can now be bounded as follows:

\begin{align*}
\left( \frac{\beta}{h_0}\right)^{\epsilon (k - K_r)} \cdot \left( \frac{\beta}{\ell_0}\right)^{(1-\epsilon) (k - K_r)} &= \\
\left( \frac{c_2 k^3}{\epsilon (k - K_r)\log^{1/5}k} \right)^{\epsilon (k - K_r)} \cdot \left( \frac{c_2 k^3}{(1-\epsilon) (k - K_r)\log ^{1/5}k}\right)^{(1-\epsilon) (k - K_r)} &= \\
\left( \frac{c_2 k^3}{\epsilon^\epsilon (1-\epsilon)^{1-\epsilon}  (k - K_r)\log ^{1/5}k}\right)^{(k - K_r)} \leq \left( \frac{2c_2k^3}{(k - K_r)\log ^{1/5}k}\right)^{(k - K_r)} 
\end{align*}

The Inclusion-Exclusion part runs in $2^{\tfrac 12 K_r \log k} = \sqrt{k}^{K_r}$.
By the same arguments as before, we can bound the total running time as desired.
\end{proof}

\begin{lemma}
The polynomial time dependency of \algo{PlaneCover} is $\BigO{k^4\log^4 k}$ and its space complexity is $\BigO{k^6\log^2 k}$ bits.
\end{lemma}
\begin{proof}
\algo{Inclusion-Exclusion} runs in $\BigO{n^{d+1}2^n}$ time when $n\leq k\log k$; the polynomial dependency is $\BigO{(k\log k)^4}$. At any point there are at most $\binom{n}{d}$ candidates, so any internal node stores a set of at most $\BigO{(k^2)^3}$ elements. There are at most $\BigO{\log k}$ such sets in memory at any time so $\BigO{k^6\log^2 k}$ bits are enough to store them.
\end{proof}

\section{Discussion}
We have presented a general algorithm that improves upon previous best algorithms for
all variations of \prob{Curve Cover} as well as for the \prob{Hyperplane Cover}
problem in $\reals^3$. 
Given good incidence bounds it should not be difficult to apply this algorithm to more geometric covering problems. 
However, such bounds are difficult to obtain in higher dimensions and for \prob{Hyperplane
Cover} the bound $\BigO{n^d/\gamma^3}$ is tight when no constraints are placed
on the input, but it is too weak to be used even in $\reals^3$. The bound by
Elekes and T\'oth works when the hyperplanes are well saturated, but the
convenient relationship between saturation and degeneracy on hyperplanes does
not extend past the $\reals^3$ setting.
Our hyperplane kernel guarantees a bound on the number of points on any $j$-flat. This overcomes the worst-case constructions for known incidence bounds, which involve placing very many points on the same line. An incidence bound for a kernelized point set might provide the needed foundation for similar \prob{Hyperplane Cover} algorithms in higher dimensions.

\bibliography{cites}{}
\bibliographystyle{abbrv}

\newpage
\appendix
\section{Algorithms}

\begin{algorithm}
\caption{Recursive Curve Cover}
\label{alg:rec-cc}
\begin{algorithmic}[1]
\Procedure{CC-Recursive}{$P,\tuple{k_1,\dots,k_r},i$}
        \If{$|P| > K_isk/2^{i-1}$}
                \State \Return{no}
        \EndIf
        \If{$|P| < K_i\log k$}
                \State \Return{\Call{Inclusion-Exclusion}{$P,k'$}}
        \EndIf
        \State let $S$ be the set of $\gamma_i$-rich-$\gamma_{i-1}$-poor candidates
        \ForAll{$S'$ s.t.\ $S'\subseteq S,|S|=k_i$}
                \If{\Call{CC-Recursive}{$P\setminus (\bigcup S'),\tuple{k_1,\dots,k_r},i+1$}}
                        \State \Return{yes}
                \EndIf
        \EndFor
        \State \Return{no}
\EndProcedure
\end{algorithmic}
\end{algorithm}
\begin{algorithm}
\caption{Recursive Plane Cover}
\label{alg:rec-pc}
\begin{algorithmic}[1]
\Procedure{PC-Recursive}{$P,\LL,\tuple{h_1,\ell_1,\dots,h_r,\ell_r},i$}
        \If{$|P| > (K_i+|\LL|)\gamma_{i-1}$}
                \State \Return{no}
        \EndIf
        \If{$|P| < K_i\log k$}
                \State \Return{\Call{Inclusion-Exclusion}{$P \cup \LL,K_i+|\LL|$}}
        \EndIf
        \State let $A$ be the set of lines in $\LL$ added at depth $\left\lceil\frac{5(i-1)-2\log k}{4}\right\rceil$ or earlier
        \If{$|A|\geq 1$}
                \ForAll{sets of $|A|$ planes $\mathcal{H}$ s.t.\ each $h\in\mathcal{H}$ covers a $\ell\in A$ and a $p\in P$ }
                        \If{\Call{PC-Recursive}{$P\setminus(\bigcup \mathcal{H}),\LL\setminus A,\tuple{h_1,\ell_1,\dots,h_r,\ell_r},i$}}
                                \State \Return{yes}
                        \EndIf
                \EndFor
                \State \Return{no}      
       \EndIf
        \State let $H$ be the set of not-too-degenerate planes
        \State let $L$ be the set of $\gamma_i\delta_i$-rich $\gamma_{i-1}$-poor lines
                                \ForAll{$\tuple{H',L'}$ s.t.\ $H'\subseteq H, |H'|=h_i$ and $L'\subseteq L, |L'|=\ell_i$}
                \State let $P'$ be the set of points that are in $P$ but not on any $h\in H'$ or on any $\ell \in L'$
          \If{\Call{PC-Recursive}{$P',\LL \cup L',\tuple{h_1,\ell_1,\dots,h_r,\ell_r},i+1$}}
                \State \Return{yes}
          \EndIf
        \EndFor
        \State \Return{no}
\EndProcedure
\end{algorithmic}
\end{algorithm}

\newpage
\section{Proof of \rlem{leafbound}}
\label{sec:longproof}
This entire section is used to prove \rlem{leafbound}. To do so, we must first give a number of auxiliary lemmas.
The whole setup of the algorithm is to be able to use the incidence bound from \rlem{numberofcurves} to bound the number of candidates at recursion level $i$.
With a bound on the number of candidates we can bound the branching at level $i$ as follows.

\begin{lemma}
\label{lem:candidatebound}
For some constant $c_1 = c_1(d, s)$ and $\alpha = 2^{d-1}$. The branching factor of an internal node of $\mathcal{T}$ at level $i$ is bounded by
$\left ( \frac{\alpha^i c_1 k}{k_i}\right)^{k_i}$.
\end{lemma}
\begin{proof}
Let the budget partition $\tuple{k_1,\dots,k_r}$ be fixed and consider recursion level $i$.
At this point at most $K_i\gamma_{i-1}\leq sk^2/2^{i-1}$ points remain and all candidate curves in $S$ are $\gamma_i$-rich.
By \rlem{numberofcurves}, $|S|$ is bounded by one of the following:
\begin{align*}
\BigO{\frac{(K_i\gamma_i)^d}{\gamma_i^{2d-1}}} &= \BigO{{\left(\frac{sk^2}{2^{i-1}}\right)}^d {\left(\frac{sk}{2^i}\right)}^{-(2d-1)}} &=& ~ \BigO{\alpha^is^{1-d}2^d k} &\\
\BigO{\frac{(K_i\gamma_i)}{\gamma_i^{2d-1}}} &=
\BigO{\left(\frac{sk^2}{2^{i-1}}\right) {\left(\frac{sk}{2^i}\right)}^{-1}} &=& ~ \BigO{k}&
\end{align*}
Let $c_1$ be the smallest constant (dependent on the constants $s$ and $d$) such that $\alpha^i\frac{c_1}{e} k$ is always greater than the implicit functions of both bounds.
At level $i$, the algorithm will branch on all possible ways of picking $k_i$ curves out of $|S|\leq \alpha^i \frac{c_1}{e} k$ candidates.
We can bound this by
\begin{align*}
\binom{\alpha^i \frac{c_1}{e} k}{k_i} \leq \left ( \frac{e \alpha^i \frac{c_1}{e} k}{k_i}\right)^{k_i} = \left ( \frac{\alpha^i c_1 k}{k_i}\right)^{k_i}
\end{align*}
\end{proof}

For the worst case analysis, we need to know for which budget partition the product of branching factors is maximized.
We therefore prove the following.

\begin{lemma}
\label{lem:manybinomprod}
Let $k_1,\dots,k_t$ be non-negative integers with a fixed sum, and $\alpha > 1$ and $\beta$ constant real numbers.
It holds that:
\[\Pi = \prod_{i=1}^t \left( \frac{\alpha^i \beta}{k_i}\right)^{k_i} \leq \left( \frac{\beta}{k_0}\right)^{\sum_{i=1}^t k_i}~ \text{ where } k_0 = \frac{\alpha - 1}{\alpha^t - \alpha}\sum_{i=1}^t k_i\]
\end{lemma}
\begin{proof}
Let $\sum_{i=1}^t k_i = k$ and assume that $k_1,...,k_t$ maximize $\Pi$. To prove the statement we explicitly compute the value of $k_i$ as a function of $\alpha$, $i$, $r$, and $k$.
Let $k_i$ and $k_{i+1}$ sum to $\kappa$, and let $c$ be any constant. Consider the function $f : [0,\kappa] \mapsto \reals$, where $f(x) = \big(\frac{c}{x}\big)^x\big(\frac{\alpha c}{\kappa-x}\big)^{\kappa-x}$.
Because $k_i$ and $k_{i+1}$ maximize $\Pi$, the function $f$ is maximal at $f(k_i)$ and thus we derive the maximum of $f$ by finding the maximum of its derivative.

\begin{align*}
\log f(x) &= x (\log c - \log x) + (\kappa-x) (\log \alpha c - \log (\kappa - x))\\
(\log f(a))' &= (\log c - \log x) - 1 - (\log \alpha c - \log (\kappa - x)) + 1\\
             &= \log \frac{\kappa - x}{x\alpha} = 0
\end{align*}

From this we derive that $f(x)$ is maximal when $\kappa - x = x\alpha$, and thus that $\alpha k_i = k_{i+1}$.
Since the above argument holds for all $i\geq 1$, only $k_1$ can be freely set and all other $k_i$ are of the form $\alpha^{i-1} k_1$.
Define $k_0 = \alpha k_1$, so that for all $i \geq 1$ we have $k_i = \alpha^i k_0$.
The expression for $k_0$ as it appears in the lemma can be derived from $\sum_{i=1}^{r}\alpha^i k_0 = k$ by finding the correct geometric series.
Filling in the computed values for $k_i$ in the definition of $\Pi$ we get:
\[\Pi = \prod_{i=1}^t \left( \frac{\alpha^i \beta}{\alpha^i k_0}\right)^{\alpha^i k_0} = \prod_{i=1}^t \left( \frac{\beta}{k_0}\right)^{\alpha^i k_0} = \left( \frac{\beta}{k_0}\right)^{\sum_{i=1}^t k_i}.\]
\end{proof}

The product $\Pi$ here is essentially the bound on $L$ from \rlem{leafbound} that we are looking for.
Before we can apply it however, we need to solve for $k_0$.
Note that $k_0$ depends on $r$ (the deepest recursion level of the algorithm) and the sum $\sum_{i=1}^r k_i$ (i.e. the budget used in the recursive part).
To determine the deepest recursion level, recall that the algorithm keeps recursing until either too few or too many points remain.
That means that we can derive the maximal recursion depth by solving what the recursion depth is where both those bounds are equal (i.e. no more branching can occur).
The algorithm switches no later than depth $r$, where at most $\frac{(d-1)}{2}K_r\log k$ points remain.
Conversely, if the instance was not immediately rejected, at most $K_tsk/2^{r-1}$ points remain.
By solving for $r$ and using the expression for $k_0$ from \rlem{manybinomprod} we can prove the following.

\begin{lemma}
\label{lem:kzero}
Let $r$ be the deepest level of recursion in \algo{CC-Recursive}, $k_0$ is bounded by:
\[k_0 = \BigO{\frac{(k - k_r)((d-1)\log k)^{d-1}}{(2sk)^{d-1}}}.\]
\end{lemma}
\begin{proof}
The algorithm does not recurse if either the number of points left is less than $\frac{(d-1)}{2}k_r\log k$, or more than $k_rsk/2^{r-1}$.
This means that it cannot recurse if \[\frac{(d-1)}{2}k_r\log k > k_rsk/2^{r-1}\]
Setting these quantities equal and solving for $r$ will thus give an upper bound on the recursion depth of any branch.
\begin{align*}
\frac{(d-1)}{2}k_r\log k &= k_rsk/2^{r-1}\\
2^r &= \frac{4sk}{(d-1)\log k} \\
r &= \log \frac{k}{\log k} + \log \frac{4s}{d-1}
\end{align*}
Thus at most the budgets $k_1,...,k_{r-1}$ summing to $k - k_r$ can be used by the recursive part of the algorithm.
Plugging these values into \rlem{manybinomprod} yields:
\[k_0 = \frac{k - k_r}{\sum_{i=1}^{r-1}\alpha^i} = \frac{(k- k_r)(\alpha - 1)}{\alpha^r - \alpha}.\]
We now expand $(2^{d-1})^r$
\begin{align*}
(2^{d-1})^r = (2^{d-1})^{\log \frac{k}{\log k} + \log \frac{4s}{d-1}} = \left(\frac{4sk}{(d-1)\log k}\right)^{d-1} = 2^{d-1}\left(\frac{2sk}{(d-1)\log k}\right)^{d-1}.
\end{align*}
We substitute $\alpha^r = (2^{d-1})^r$ in the expression for $k_0$ to get:
\begin{align*}
k_0 = \frac{(k- k_r)(2^{d-1} - 1)}{2^{d-1}\left(\frac{2sk}{(d-1)\log k}\right)^{d-1} - 2^{d-1}} =
\BigO[\Theta]{\frac{(k- k_r)}{\left(\frac{2sk}{(d-1)\log k}\right)^{d-1} - 1}} = 
\BigO[\Omega]{\frac{(k- k_r)}{\left(\frac{2sk}{(d-1)\log k}\right)^{d-1}}}.
\end{align*}
By simplifying the last expression the proof is complete.
\end{proof}

We now have enough machinery to prove \rlem{leafbound}.

\begin{repeatlemma}{lem:leafbound}
Let $L$ be the number of leaves in $\mathcal{T}_j$.
Then for some constant $c_2=c_2(d,s)$, $L$ is bounded by
\[ L \leq \left(\frac{c_2k^d}{(k-k_r)\log^{d-1} k}\right)^{K_j-k_r}\]
\end{repeatlemma}
\begin{proof}[Proof of \rlem{leafbound}]
Let $\alpha = 2^{d-1}$.
\rlem{candidatebound} gives a bound of $\left ( \frac{\alpha^i c_1 k}{k_i}\right)^{k_i}$ on the branching of an internal node at recursion level $i$.
Taking the product of branching factors on the recursion levels $j$ through $r$ gives a bound on $T_j$.
We can directly apply the bound from \rlem{manybinomprod} to bound this product and therefore bound $T_j$.
\begin{align*}
T_j\leq \prod_{i=j}^{r-1} \left(\frac{\alpha^i c_1 k}{k_i}\right)^{k_i} \leq
\left(\frac{c_1k}{k_0}\right)^{\sum_{i=j}^{r-1}k_i} =
\left(\frac{c_1k}{k_0}\right)^{K_j-k_r}
\end{align*}
Now substitute $k_0$ from \rlem{kzero} and collect any constants in $c_2 = c_2(d, s)$ to get
\begin{align*}
T \leq \left( \frac{c_1 k (2sk)^{d-1}}{(k - k_r)((d-1)\log k)^{d-1}}\right)^{K_j - k_r} = {\left(\frac{c_2k^d}{(k-k_r)\log^{d-1} k}\right)}^{K_j-k_r}.
\end{align*}
\end{proof}

\end{document}